\documentclass[11pt]{article}

\usepackage{amssymb}
\usepackage{amsfonts}
\usepackage{graphicx}
\usepackage{latexsym}
\usepackage{amsmath}
\usepackage{changepage}
\usepackage{color}
\usepackage{latexsym}
\usepackage{tablefootnote}
\usepackage[shortlabels]{enumitem}
\usepackage{epsfig}
\usepackage{multirow}
\usepackage{float}
\usepackage{array}
\usepackage{bbm}
\usepackage{dsfont}
\usepackage{url}

\allowdisplaybreaks

\newtheorem{thm}{Theorem}[section]
\newtheorem{theorem}[thm]{Theorem}

\newtheorem{lemma}[thm]{Lemma}

\newtheorem{corollary}[thm]{Corollary}

\newtheorem{remark}[thm]{Remark}

\newtheorem{assumption}[thm]{Assumption}
\newtheorem{algorithm}[thm]{Algorithm}

\begin{document}
\title{Reduction from the Partition Problem: Dynamic Lot Sizing Problem with Polynomial Complexity}

\author{{\bf{Chee-Khian Sim}}\footnote{Email address:
chee-khian.sim@port.ac.uk}
 \\ School of Mathematics and Physics\\
University of Portsmouth \\ Lion Gate Building, Lion Terrace \\ Portsmouth PO1 3HF \\ United Kingdom
}

\date{Last updated:  23 December 2025}

\maketitle
\begin{abstract}
In this note, we polynomially reduce an instance of the partition problem to a dynamic lot sizing problem, and show that solving the latter problem solves the former problem.   By solving the dynamic program formulation of the dynamic lot sizing problem, we show that the instance of the partition problem can be solved with pseudo-polynomial time complexity.  Numerical results on solving instances of the partition problem are also provided using an implementation of the algorithm that solves the dynamic program.  We conclude by discussing polynomial time solvability of the partition problem through further observation on the dynamic program formulation of the dynamic lot sizing problem.
 
\vspace{15pt}

\noindent {\bf{Keywords}}. Partition problem; dynamic lot sizing model; dynamic program; pseudo-polynomial time complexity; polynomial time complexity. 

\end{abstract}

\section{Introduction}\label{sec:Intro}

An NP-complete problem is in the class of NP, and at the same time, it is NP-hard.    In this note, we consider solving a well-known NP-complete problem -  the partition problem \cite{garey,Karp,Lenstra,Pinedo}.  We relate the problem to a dynamic lot sizing problem, and by solving the dynamic program formulation of the latter problem, we show that we can solve any instance of the partition problem with pseudo-polynomial time complexity.  We conclude this note by discussing the possibility of polynomial time complexity of the partition problem through our approach.   Recall from \cite{garey} that an algorithm for a problem is a pseudo-polynomial time algorithm for the problem if for a problem instance $I$ of the problem, its time complexity function is bounded above by a polynomial function of two variables Length$[I]$ and Max$[I]$, where Length$[I]$ is an integer that corresponds to the number of symbols used to describe $I$ under some reasonable encoding scheme for the problem, while Max$[I]$ is an integer that corresponds to the magnitude of the largest number in $I$.   On the other hand, an algorithm for a problem is a polynomial time algorithm for the problem if for a problem instance $I$ of the problem, its time complexity function is bounded above by a polynomial function of Length$[I]$.   

\subsection{Notations}\label{subsec:notations}

Let $f(x)$ and $g(x)$ be two nonnegative real-valued functions, where $x \in \mathbb{Z}^{++}$.   We write $g(x) = O(f(x))$ to mean that $g(x) \leq K f(x)$ for some positive constant $K$ and all $x > 0$.

\vspace{10pt}

\noindent Furthermore, $\mathbb{I}(x)$, where $x \in \mathbb{Z}$, is defined to be $1$ for $ x > 0$, and $0$ otherwise; and $x^+ = \max \{x, 0\}$, where $x \in \Re$.   Also, we have $\min \{x,y\} = x - (x-y)^+$ for $x,y \in \Re$.

%

\section{The Partition Problem and its Reduction to a Dynamic Lot Sizing Problem}\label{sec:Reduction}

Let $S = \{1, \ldots, n \}$ and $a_i \in \mathbb{Z}^{++}$ for $i \in S$, with $\sum_{i \in S} a_i = 2C$.  The partition problem is to find a subset $A$ of $S$ such that 
\begin{eqnarray*}
\sum_{i \in A} a_i = \sum_{i \in S \backslash A} a_i = C.
\end{eqnarray*}

\vspace{10pt}

\noindent It is known that the partition problem is NP-complete \cite{Cook,garey,Karp}.  We call an instance of the partition problem {\bf{PPi}}.

\vspace{10pt}

\noindent In this section, we reduce {\bf{PPi}} to a dynamic lot sizing problem in polynomial time and show that solving the dynamic lot sizing problem solves {\bf{PPi}} in Theorem \ref{thm:solution}.  As a consequence, if the dynamic lot sizing problem can be solved with polynomial complexity in some sense, {\bf{PPi}} can also be solved with polynomial complexity in the same sense.

\vspace{10pt}

\noindent Dynamic lot sizing problem is introduced in \cite{Wagner2}, and has since been studied intensively by researchers.  We consider a variant of this basic problem which is related to remanufacturing. 

\vspace{10pt}

\noindent In the following, we list down the parameters of the dynamic lot sizing model that we are considering in this note:

\vspace{10pt}

\noindent {\bf{Parameters}}:
\begin{itemize}
\item $N = $ number of periods in the time horizon, where $N \geq 1$;
\item $D_i = $ demand for serviceable products in the $i^{th}$ period, where  $i = 1, \ldots, N$.  We assume that $D_i \in \mathbb{Z}^{++}, i = 1, \ldots, N$;
\item $R_i = $ returned products as cores at the beginning of the $i^{th}$ period, where $i = 1, \ldots, N$.  We assume that $R_i \in \mathbb{Z}^+, i = 1, \ldots, N$, and set $R_{N+1} = 0$;
\item $K_{r,i} = $ setup cost when there is remanufacturing at the beginning of the $i^{th}$ period, where $i = 1, \ldots, N$.  We let  $K_{r,i} \geq 0$ for all $i = 1, \ldots, N$, and set $K_{r,N+1} = 0$;
\item $\Delta K_{m,i} = $ setup cost when there is manufacturing at the beginning of the $i^{th}$ period, where  $i = 1, \ldots, N$.  We let $\Delta K_{m,i} \geq 0$ for all $i = 1, \ldots, N$, and set $\Delta K_{m,N+1} = 0$;
\item $h_{s,i} = $ unit holding cost of serviceable product over the $i^{th}$ period whether from manufacturing or remanufacturing, where $i = 1, \ldots, N$.  We let $h_{s,i} \geq 0$ for all $i = 1, \ldots, N$;
\item $h_{c,i} = $ unit holding cost of core over the $i^{th}$ period, where $i = 1, \ldots, N$.  We let $h_{c,i} \geq 0$ for all $i = 1, \ldots, N$;
\item $c_{r,i} = $ unit remanufacturing cost in the $i^{th}$ period, where $i = 1, \ldots, N$.  We let $c_{r,i} \geq 0$ for all $i = 2, \ldots, N$, and set $c_{r,N+1} = 0$;
\item $c_{m,i} = $ unit manufacturing cost in the $i^{th}$ period, where $i = 1, \ldots, N$.  We let $c_{m,i} \geq 0$ for all $i = 1, \ldots, N$, and set $c_{m,N+1} = 0$. 
\end{itemize}
 
\noindent We further impose assumptions on the above parameters as follows:
\begin{assumption}\label{ass:parameters}
\begin{enumerate}[(a)]
\item $h_{s,i} + c_{m,i} > \Delta K_{m,i+1} + c_{m,i+1}$ for $1 \leq i \leq N$;
\item $h_{s,i} + c_{r,i}  > K_{r,i+1} + h_{c,i} + c_{r,i+1}$ for $1 \leq i \leq N$.
\end{enumerate}
\end{assumption}

\noindent These assumptions are crucial to prove Lemma \ref{lem:optimalproperties}, which in turn is needed to formulate the dynamic program for the dynamic lot sizing problem we are considering in this note, and also to solve it efficiently.


\vspace{10pt}

\noindent Demand must be satisfied in each period in our model. Our objective for the model is to minimize its total cost, which comprises of setup costs for manufacturing and remanufacturing, holding costs for serviceable products and cores, manufacturing and remanufacturing costs.  We have the following sequence of events in our model - at the beginning of a period, (i) returned products arrive as cores; (ii) number of units of serviceable products to produce through remanufacturing and manufacturing is determined; (iii) demand in the period is satisfied; (iv) any leftover cores and serviceable products are held to the next period.

%

\vspace{10pt}

\noindent The dynamic lot sizing problem ({\bf{DLSP}}) we are considering is given by:
\begin{eqnarray*}
\min \sum_{i=1}^N (K_{r,i} \mathbb{I}(x_i) + \Delta K_{m,i} \mathbb{I}(y_i) + c_{r,i} x_i + c_{m,i} y_i + h_{c,i}[ J_{i} - x_i] + h_{s,i} I_{i+1})
\end{eqnarray*}
subject to
\begin{eqnarray*}
& & J_{i+1} = J_i + R_{i+1} - x_i,\ i = 1, \ldots, N, \\
& & I_{i+1} = I_i + x_i + y_i - D_i, \ i = 1, \ldots, N, \\
& & x_i \leq J_i, \ i = 1, \ldots, N, \\
& & J_i, I_i \geq 0, \ i = 2, \ldots, N+1, \\
& & x_i, y_i \in \mathbb{Z}^+, \ i = 1, \ldots, N, \\
& & J_1 = R_1,  I_1 = 0.
\end{eqnarray*}
The decision variables in the above minimization problem are:
\begin{itemize}
\item $x_i = $ number of units of cores remanufactured in the $i^{th}$ period;
\item $y_i = $ number of units of serviceable products obtained by manufacturing in the $i^{th}$ period,
\end{itemize}
while 
\begin{itemize}
\item $J_i = $ number of units of cores at the beginning of the $i^{th}$ period;
\item $I_i = $ number of units of available serviceable products at the beginning of the $i^{th}$ period.
\end{itemize}
The objective function in the above minimization problem is the total cost of the model.  The first constraint tells us the number of units of cores available at the beginning of the $(i+1)^{th}$ period, $i = 1, \ldots, N$, after events occurred in the $i^{th}$ period.  The second constraint tells us the number of units of serviceable products available at the beginning of the $(i+1)^{th}$ period, $i = 1, \ldots, N$, after events occurred in the $i^{th}$ period.  The third constraint tells us that the number of cores remanufactured in the $i^{th}$ period cannot exceed the cores available in the period.  The fourth constraint tells us that the number of units of cores and serviceable products at the beginning of the $i^{th}$ period are never negative.  The next constraint is the sign constraint on the decision variables in the problem, while the last constraint sets specific values on $J_1, I_1$.    Note that the parameters in {\bf{DLSP}} satisfy Assumption \ref{ass:parameters}.

\vspace{10pt}

\noindent Let us call the optimal value of the minimization problem $C^\ast$, and its optimal solution $x_i^\ast, y_i^\ast$, $i = 1, \ldots, N$, with $J_i^\ast = J_{i-1}^\ast + R_{i} - x_{i-1}^\ast, I_i^\ast = I_{i-1}^\ast + x_{i-1}^\ast + y_{i-1}^\ast - D_{i-1}$, $i = 2, \ldots, N+1$, $J_1^\ast = R_1, I_1^\ast = 0$.

\vspace{10pt}

\noindent We reduce {\bf{PPi}} in polynomial time to the above dynamic lot sizing problem by setting appropriate values for parameters of the model as follows:
\begin{itemize}
\item $N =n$;
\item $D_i = a_{i}$, $i = 1, \ldots, N (= n)$;
\item $R_1 = C$, $R_i = 0$, $i = 2, \ldots, N$;
\item $K_{r,i} = \Delta K_{m,i} = 1$, $i =1, \ldots, N$;
\item $h_{s,i} = 3$, $i = 1, \ldots, N$;  
\item $h_{c,i} = 0$, $i = 1, \ldots, N$;
\item $c_{r,i} = 0$, $i = 1, \ldots, N$;
\item $c_{m,i} = 1$, $i = 1, \ldots, N$.
\end{itemize}
It is easy to check that parameters of the model with the above values satisfy Assumption \ref{ass:parameters}.  We call the dynamic lot sizing problem with these values for its parameters {\bf{DLSPp}}, and this problem is a special case of {\bf{DLSP}}.  We have the following theorem:

\begin{theorem}\label{thm:solution}
{\bf{PPi}} can be solved by solving {\bf{DLSPp}}.
\end{theorem}
\begin{proof}
{\it{Claim 1:}} Suppose there exists a subset $A$ of $S = \{1, \ldots, n\}$ such that 
\begin{eqnarray*}
\sum_{i \in A} a_i = \sum_{i \in S \backslash A} a_i = C,
\end{eqnarray*}
then the optimal value to {\bf{DLSPp}} is at most $N+C$.

\vspace{5pt}

\noindent  It is easy to see that by remanufacturing $D_i$ units of cores in the $i^{th}$ period, when $i \in A$, and manufacturing $D_i$ units from raw materials in the $i^{th}$ period, when $i \not\in A$,  total cost is $N + C$, and it is feasible to {\bf{DLSPp}}.  Hence, the optimal value to {\bf{DLSPp}} is at most $N + C$.

\vspace{5pt}

\noindent {\it{Claim 2:}}  Suppose the optimal value to {\bf{DLSPp}} is at most $N+C$.  Let $A$ contains elements $i \in S = \{ 1, \ldots, N\}$ such that we remanufacture in the $i^{th}$ period in {\bf{DLSPp}}.  Then we have
\begin{eqnarray*}
\sum_{i \in A} a_i  = \sum_{i \in S \backslash A} a_i = C.
\end{eqnarray*}

\noindent First note that under optimality, whenever we produce, we only produce enough to satisfy demand for the period, and do not hold serviceable products to the next period.  To see this, suppose we hold a serviceable product to the next period, then a cost of $h_{s,i} = 3$ is incurred.  If we do not produce the serviceable product in the current period, but in the next period, we do not incur the holding cost of $h_{s,i} = 3$ and may even save on its manufacturing cost if this product is obtained by manufacturing, but we incur a possible setup cost of $1$ due to remanufacturing or manufacturing, and possible unit manufacturing cost $c_{m,i+1} = 1$ in the next period.  In the new setup, total cost is reduced by at least $1$, but this contradicts optimality.  Hence, under optimality, whenever we produce, we only produce enough to satisfy demand for the period, and do not hold serviceable products to the next period.   It is easy to see that all $R_1 = C$ units of cores are remanufactured to satisfy demand since there is no cost for remanufacturing.  Note that these cores need not be all remanufactured in the $1^{st}$ period and they can be held to later periods for remanufacturing without incurring holding cost since $h_{c,i} = 0$.  Now, total demand is $\sum_{i=1}^N D_i = \sum_{i=1}^n a_i = 2C$, and since half of these demands is satisfied through remanfacturing and that all demand has to be satisfied, the other half of these demands has to be satisfied through manufacturing, incurring a total manufacturing cost of $C$, since $c_{m,i} = 1$.  In each period, we always have manufacturing and/or remanufacturing to satisfy demand in the period, as we do not have serviceable products held from earlier periods to satisfy demand in the period.  Total setup cost is then at least $N$.   Hence, total cost is at least $N + C$.  However, the optimal value to {\bf{DLSPp}} is at most $N+C$. Therefore, under optimality, we must have total cost is exactly $N + C$, leading to total setup cost to be exactly $N$, and we either remanufacture or manufacture in a period.  Claim 2 then follows.  
\end{proof}

\vspace{10pt}

\noindent Note that {\bf{DLSPp}} and Theorem \ref{thm:solution} with the claims in its proof follow \cite{vandenHeuvel}, while the proof of Claim 2 in the theorem is inspired by \cite{vandenHeuvel}.

\section{Dynamic  Program Formulation of Dynamic Lot Sizing Problem and its Solution}\label{sec:DynamicprogramSolution}

\noindent We propose a dynamic program formulation of {\bf{DLSP}} in this section.  Before we do this, we state and prove the following lemma that is the key which allows us to have the formulation and then solving it efficiently.

\begin{lemma}\label{lem:optimalproperties}
In {\bf{DLSP}}, suppose we produce in the $i^{th}$ period, where $1 \leq i \leq N-1$, then $I^\ast_{i+1} = 0$. 
\end{lemma}
\begin{proof}
\noindent We show that $I^{\ast}_{i+1} = 0$ by assuming that $I^\ast_{i+1} \geq 1$, and show that this leads to a contradiction.  Since we produce in the $i^{th}$ period, we have manufacturing or remanufacturing or both in the $i^{th}$ period,  that is, $x_i^\ast + y_i^\ast \geq 1$.   Suppose we have remanufacturing in the $i^{th}$ period, that is, $x_i^\ast \geq 1$.   By reducing remanufacturing by 1 unit, noting that demand in period is still satisfied since we assume that $I^\ast_{i+1} \geq 1$, we have a cost reduction of at least $c_{r,i} + h_{s,i}$, but we incur an additional holding cost of a unit of core of $h_{c,i}$.  The unit of serviceable product can be ``reinstated" through remanufacturing in the $(i+1)^{th}$ period by incurring a cost of at most $K_{r,i+1} + c_{r,i+1}$.  In this case, it is easy to see that the total cost is reduced by at least $c_{r,i} + h_{s,i} -  h_{c,i} - K_{r,i+1} - c_{r,i+1}$ which is positive by Assumption \ref{ass:parameters}(b).  This is a contradiction to optimality. We have a similar argument to show contradiction if we have manufacturing in the $i^{th}$ period using Assumption \ref{ass:parameters}(a). Therefore, we show that $I^\ast_{i+1} = 0$. 
\end{proof}

\vspace{10pt}

\noindent The results in the above lemma is a strong version of the well-known zero-inventory property of the dynamic lot sizing problem, which first appeared in \cite{Wagner2}.  It says that under optimality, if we produce in the current period, then the optimal inventory policy is to have no serviceable product available at the beginning of the next period.  That is, we only produce enough to satisfy demand in the current period.

\begin{corollary}\label{cor:optimalproperties}
In {\bf{DLSP}}, suppose we produce in the $i^{th}$ period for some $i$, $1 \leq i \leq N-1$.  Suppose further that $I_i^\ast = 0$.  Then, $x_i^\ast + y_i^\ast = D_i$.
\end{corollary}
\begin{proof}
By Lemma \ref{lem:optimalproperties}, we have $I_{i+1}^\ast = 0$.  The result then follows by observing that $I_{i+1}^\ast = I_i^\ast + x_i^\ast + y_i^\ast - D_i$.
\end{proof}

\begin{remark}\label{rem:optimalproperties}
Under optimality, it is easy to convince ourselves from the proof of Lemma \ref{lem:optimalproperties} that the results in the lemma and Corollary \ref{cor:optimalproperties} still hold if  we let $i = N$ in their statements.  Furthermore, since $D_i > 0$ for $i = 1, \ldots, N$ and $I_1^\ast = I_1 = 0$, Lemma \ref{lem:optimalproperties} and Corollary \ref{cor:optimalproperties} imply that we have $I^\ast_i = 0$,  $y_{i}^\ast + x_i^\ast = D_i$ for $i = 1, \ldots, N$.  This means that in each period, we always produce and only enough to satisfy demand for the period. 
\end{remark}  

\noindent Let us now consider a dynamic program which we use to solve {\bf{DLSP}}.  The design of the dynamic program is motivated by Lemma \ref{lem:optimalproperties},  Corollary \ref{cor:optimalproperties} and Remark \ref{rem:optimalproperties}. 

\vspace{10pt}

\noindent For $i = 1, \ldots, N$, and $J_i \geq 0$,
\begin{eqnarray}\label{eq:recursive}
C_i^{\ast\ast}(J_i) & := & \min \{ K_{r,i} \mathbb{I}(x_i)+ \Delta K_{m,i} \mathbb{I}(D_i - x_i) + h_{c,i} [J_i - x_i] + c_{r,i} x_i  +  c_{m,i}[D_i  - x_i] + \nonumber \\ 
 & &  C_{i+1}^{\ast\ast}(J_i - x_i + R_{i+1}) \  {\mbox{\Large$|$}}\    x_i \leq J_i,  x_i \leq D_i, x_i \in \mathbb{Z}^+ \}
\end{eqnarray}
\noindent We have the convention that $C_{N+1}^{\ast\ast}(J_{N+1}) = 0$ for all $J_{N+1} \geq 0$.

\vspace{10pt}

\noindent The first term within the minimization in (\ref{eq:recursive}) can be interpreted as  the setup cost for remanufacturing in the $i^{th}$ period; the second term is the setup cost for manufacturing in the $i^{th}$ period;  the third term is the holding cost during the $i^{th}$ period for cores not remanufactured in the period; the fourth term is the remanufacturing cost for serviceable product in the $i^{th}$ period to satisfy demand in the period; the fifth term is the manufacturing cost for serviceable products in the $i^{th}$ period to satisfy demand in the period; the last term can be interpreted as the optimal cost from the $(i+1)^{th}$ period up to the end of the time horizon.

%
%
%
%
 
 \vspace{10pt}
 
 \noindent The following lemma relates the above dynamic program to {\bf{DLSP}}:
\begin{lemma}\label{lem:relation}
We have $C^\ast = C_1^{\ast\ast}(R_1)$.
\end{lemma}
\begin{proof}
We are given an optimal solution $x_i^\ast, y_i^\ast, i = 1, \ldots, N$, to {\bf{DLSP}}.   We have $I_{1}^\ast = 0$ and by Remark \ref{rem:optimalproperties}, $I_{i}^\ast = 0$ for   $i =2 \ldots, N$  and  $x_i^\ast + y_i^\ast = D_{i}$ for $i = 1, \ldots, N$.   We also have $J_{i+1}^\ast = J_{i}^\ast - x_{i}^\ast + R_{i+1}$, $i =1, \ldots, N$, where $J_{l}^\ast = R_{1}$.  We see that $x^\ast_{i}, i = 1, \ldots, N$ is a feasible solution to the dynamic program (\ref{eq:recursive}), where $J_{1} = R_{1}$, with its objective function value equal to $C^\ast$.  Hence, we have $C^\ast \geq C_1^{\ast\ast}(R_{1})$.  On the other hand, given an optimal solution $ x^{\ast\ast}_{i}, i = 1, \ldots, N$, to the dynamic program (\ref{eq:recursive}), where $J_{1} = R_{1}$.   It is easy to convince ourselves that $x_{i}^{\ast\ast}, D_i - x_{i}^{\ast\ast},  i = 1, \ldots, N$, is a feasible solution to {\bf{DLSP}}, since all its constraints are satisfied, and its objective function value is equal to $C_1^{\ast\ast}(R_{1})$.  Therefore $C^\ast \leq C_1^{\ast\ast}(R_{1})$.  The lemma is hence proved.
\end{proof}

\vspace{10pt}

\noindent By the above lemma, we are able to solve {\bf{DLSP}} by solving the dynamic program (\ref{eq:recursive}) with $J_1 = R_1$.   We next describe an algorithm to solve {\bf{DLSP}} by solving this dynamic program.  Before we do this, we have a lemma below that is the basis for the algorithm and further allows us to show Theorem \ref{thm:timecomplexity}:
\begin{lemma}\label{lem:dynamicprogramobservation}
When solving {\bf{DLSP}} using the dynamic program (\ref{eq:recursive}), where we set $J_1 = R_1$, for all $i = 2, \ldots, N$,  we evaluate $C_i^{\ast\ast}(J_i)$ in (\ref{eq:recursive}) for $J_i = \hat{J}_i + R_i$, where $\hat{J}_i$ takes integer value between $((\ldots((R_1-D_1)^+ + (R_2-D_2))^+ + \ldots)^+ + (R_{i-1} - D_{i-1}))^+$ and $R_1 + \ldots + R_{i-1}$ inclusively.
\end{lemma}
\begin{proof}
We prove the statement in the lemma by induction on $i = 2, \ldots, N$.   We have from (\ref{eq:recursive}) where $i = 1$  that we only need to find $C_2^{\ast\ast}(J_2)$ for $J_2 = R_1 + R_2 - x_1$, where $x_1 \leq \min \{ R_1, D_1\} = R_1 - (R_1 - D_1)^+$,  $x_1 \in \mathbb{Z}^+$.   Therefore, if we let $\hat{J}_2 = R_1 - x_1$, then we have $J_2 = \hat{J}_2 + R_2$, where $\hat{J}_2$ takes integer value between $(R_1 - D_1)^+$ and $R_1$ inclusively.   Hence, statement holds for $i = 2$.   Suppose the statement in the lemma holds for $i = i_0$, where $i_0 <N$.   We have from (\ref{eq:recursive}), where $i = i_0$, that we only need to find $C_{i_0+1}^{\ast\ast}(J_{i_0+1})$ for $J_{i_0+1} = J_{i_0} - x_{i_0} + R_{i_0+1}$, where $x_{i_0} \leq \min \{J_{i_0}, D_{i_0} \} = J_{i_0} - (J_{i_0} - D_{i_0})^+$,  $x_{i_0} \in \mathbb{Z}^+$.  Therefore, if we let $\hat{J}_{i_0+1} = J_{i_0} - x_{i_0}$, we have $J_{i_0+1} = \hat{J}_{i_0+1} + R_{i_0+1}$, where $\hat{J}_{i_0+1}$ lies between $(J_{i_0} - D_{i_0})^+$ and $J_{i_0}$.  By induction hypothesis, $J_{i_0} = \hat{J}_{i_0} + R_{i_0}$, where $\hat{J}_{i_0}$ takes integer value between  $((\ldots((R_1-D_1)^+ + (R_2-D_2))^++ \ldots)^+ + (R_{i_0-1} - D_{i_0-1}))^+$ and $R_1 + \ldots + R_{i_0-1}$ inclusively.  Since $\hat{J}_{i_0+1}$ lies between $(\hat{J}_{i_0} + R_{i_0} - D_{i_0})^+$ and $\hat{J}_{i_0} + R_{i_0}$, with  $\hat{J}_{i_0}$ taking integer value between  $((\ldots((R_1-D_1)^+ + (R_2-D_2))^+ + \ldots)^+ + (R_{i_0-1} - D_{i_0-1}))^+$ and $R_1 + \ldots + R_{i_0-1}$ inclusively, we see that the statement in the lemma holds for $i = i_0 + 1$.  Therefore, the statement in the lemma holds for all $i = 2, \ldots, N$ by induction.
\end{proof}

\begin{algorithm}\label{alg:optimaldynamicprogram}  
\ 
\begin{adjustwidth}{0.5cm}{}
{\bf{Step 1}}.  Iterate from $i = N$ to $2$, and use previously computed values for $C_{i+1}^{\ast\ast}(J_{i+1})$, with $C_{N+1}^{\ast\ast}(J_{N+1}) = 0$,  to find $C_i^{\ast\ast}(J_i)$ from (\ref{eq:recursive}) with $J_i = \hat{J}_i + R_i$ with $\hat{J}_i$ taking integer value between  $((\ldots((R_1-D_1)^+ + (R_2-D_2))^+ + \ldots)^+ + (R_{i-1} - D_{i-1}))^+$ and $R_1 + \ldots + R_{i-1}$ inclusively. 

\vspace{5pt}

\noindent {\bf{Step 2}}.  Find $C_1^{\ast\ast}(R_1)$ using (\ref{eq:recursive}), where $C_2^{\ast\ast}(\hat{J}_2 + R_2)$ in (\ref{eq:recursive}), with $\hat{J}_2$ taking integer value between $(R_1 - D_1)^+$ and $R_1$ inclusively, have been computed in Step 1.
\end{adjustwidth}
\end{algorithm}
After executing the algorithm, we can determine $x_i^\ast$ for $i = 1, \ldots, N$.   If $x_i^\ast = 0$, then we do not remanufacture when we produce, otherwise, we remanufacture when we produce.

\vspace{10pt}

\noindent Theorem \ref{thm:timecomplexity} below states the complexity to solve {\bf{DLSP}} using its dynamic program formulation.   
 
\begin{theorem}\label{thm:timecomplexity}
{\bf{DLSP}} can be solved using ${O}\left( \sum_{i=1}^N \sum_{k=L_i}^{R_1 + \ldots + R_{i-1}}(\min\{k + R_i, D_i \} + 1)^2 \right)$, where $L_i = ((\ldots((R_1-D_1)^+ + (R_2-D_2))^+ + \ldots)^+ + (R_{i-1} - D_{i-1}))^+$, multiplication, addition and comparison operations on parameters of the model.
\end{theorem}
\begin{proof}
We solve {\bf{DLSP}} using the dynamic program formulation (\ref{eq:recursive}) through Algorithm \ref{alg:optimaldynamicprogram}.   For each $i = 1, \ldots, N$, by Lemma \ref{lem:dynamicprogramobservation},   (\ref{eq:recursive}) needs to be solved for $J_i = \hat{J}_i + R_i$, where $\hat{J}_i$ runs from $((\ldots((R_1-D_1)^+ + (R_2-D_2))^+ + \ldots)^+ + (R_{i-1} - D_{i-1}))^+$ to $R_1 + \ldots + R_{i-1}$.  For each $J_i$,  there are at most  $\min\{J_i, D_i \}+1$ entries to find their mininum.   Each entry  requires $O(1)$ multiplications, $O(1)$ additions to evaluate, hence leading to a total of $O(\min \{J_i, D_i \} + 1)$ multiplications and additions for all entries.   Finding the minimum in the minimization problem can be achieved by $O((\min \{J_i, D_i \}+1)^2)$ comparison operations.  Hence, for $i = 1, \ldots, N$, for each $J_i$, solving (\ref{eq:recursive}) requires a total of $O((\min \{J_i, D_i \}+1)^2)$ multiplication, addition and comparison operations.  We have $J_i = \hat{J}_i + R_i$, where $\hat{J}_i$ runs from $((\ldots((R_1-D_1)^+ + (R_2-D_2))^+ + \ldots)^+ + (R_{i-1} - D_{i-1}))^+$ to $R_1 + \ldots +R_{i-1}$.   Therefore, for each $i = 1, \ldots, N$, the total number of operations to solve $C^{\ast\ast}_i(J_i)$ in (\ref{eq:recursive}) taking into account various values of $J_i$  is  ${O}\left(\sum_{k=L_i}^{R_1 + \ldots + R_{i-1}}(\min\{k + R_i, D_i \}+1)^2 \right)$, where $L_i = ((\ldots((R_1-D_1)^+ + (R_2-D_2))^+ + \ldots)^+ + (R_{i-1} - D_{i-1}))^+$.  Hence, we have a total  multiplication, addition and comparison operations of ${O}\left( \sum_{i=1}^N \sum_{k=L_i}^{R_1 + \ldots + R_{i-1}}(\min\{k + R_i, D_i \} + 1)^2 \right)$, where $L_i = ((\ldots((R_1-D_1)^+ + (R_2-D_2))^+ + \ldots)^+ + (R_{i-1} - D_{i-1}))^+$, summing $i$ from $1$ to $N$, to solve {\bf{DLSP}} .   \end{proof}

\vspace{10pt}

\noindent With the above theorem, we now proceed to find the time complexity to solve {\bf{PPi}}.  Let us denote $a_{\max} = \max_{1 \leq i \leq n} \{ a_i \}$.  

\begin{corollary}\label{cor:timecomplexity}
{\bf{PPi}} can be solved using ${O}(n^2a_{\max}^3)$ multiplication, addition and comparison operations on $\{a_i \ ; \ i = 1, \ldots, n \}$.
\end{corollary}
\begin{proof}
We have {\bf{DLSPp}} is a special case of {\bf{DLSP}} with $N = n$, $D_i = a_i, i = 1, \ldots, N$,  $R_1 = C$ and $R_i = 0, i = 2, \ldots, N$.  Furthermore, $\sum_{i=1}^n a_i = 2C$ which implies that $C \leq na_{\max}$.    Let us now apply these on Theorem \ref{thm:timecomplexity} to find the number of multiplication, addition and comparison operations needed to solve {\bf{PPi}}.  We have
\begin{eqnarray*}
& & \sum_{i=1}^N \sum_{k=L_i}^{R_1 + \ldots + R_{i-1}}(\min\{k + R_i, D_i \} + 1)^2 \\
& = & (\min\{ R_1, D_1 \}  + 1)^2 + \sum_{i=2}^N \sum_{k=L_i}^{R_1 + \ldots + R_{i-1}}(\min\{k + R_i, D_i \} + 1)^2 \\
& = & (\min\{ C, a_1 \}  + 1)^2 + \sum_{i=2}^n \sum_{k=L_i}^{C}(\min\{k, a_i \} + 1)^2 \\
& \leq & (a_{\max} + 1)^2 + \sum_{i=2}^n \sum_{k=L_i}^C (a_{\max} + 1)^2 \\
& \leq & (a_{\max} + 1)^2+ (n-1) C (a_{\max} + 1)^2 \\
& = & O(n C a_{\max}^2).
\end{eqnarray*}
The result in the corollary then follows from the above since  $C \leq na_{\max}$.
\end{proof}

\section{Numerical Studies}\label{sec:Numerical}

We implement Algorithm \ref{alg:optimaldynamicprogram} for {\bf{DLSP}} using Matlab R2024b, and run the resulting Matlab program on a Windows 11 desktop with 13th Gen Intel(R) Core and installed RAM of 16GB.   To simplify the implementation, we let $\hat{J}_i$ runs from 0 to $R_1 + \ldots + R_{i-1}$ instead of from $((\ldots((R_1-D_1)^+ + (R_2-D_2))^+ + \ldots)^+ + (R_{i-1} - D_{i-1}))^+$ to $R_1 + \ldots + R_{i-1}$ as in the algorithm.  We run the Matlab program on {\bf{PPi}}s in the form of {\bf{DLSP}}s in our experiments.

\vspace{10pt}

\noindent For ease in presentation, let us denote $\{ a_1, \ldots, a_n \}$ in {\bf{PPi}} by $\Omega$, where the elements in the set are ordered accordingly to their indices.  

\vspace{10pt}

\noindent In Table \ref{table1}, we report the outcome upon solving different {\bf{PPi}} by running the Matlab program which implements Algorithm \ref{alg:optimaldynamicprogram}.  It indicates that the algorithm is able to solve these {\bf{PPi}}s correctly.

\begin{table}[ht]
\centering
\begin{tabular}{||c|c|c|c|c||c|c||}  \hline \hline
  $n$  & $\Omega$                   &   $C$      &  $A$ exists?      &  $A$              &  Solution found?      & $C_1^{\ast\ast}(R_1)$ \\ \hline \hline
  3     &  $\{10, 34, 40 \}$          &   $42$    &  No                     &   -                  &  No                            &  $46$ \\ \hline
  3     &   $\{10,30, 20 \}$          &  $30$     &  Yes                     & $\{ 2 \}$    &  Yes                           &  $33$  \\ \hline
  5      & $\{10, 33, 40, 5, 8\}$   & $48$      & Yes                     & $\{ 3,5 \}$  & Yes                            & $53$ \\ \hline
  5      & $\{10, 33, 38, 5, 8 \}$  & $47$      & No                     & -                     & No                             & $53$ \\ \hline
  8      & $\{10, 33, 38, 5,$         & $208$    & Yes                     & $\{ 4,7,8\}$  &  Yes                   & $216$ \\ 
           & $50, 77, 89, 114 \}$    &                &                            &                         &                                &             \\ \hline
  8      & $\{10, 33, 38, 5,$         & $210$    & No                     &  -                      &  No                         & $219$ \\
           &  $52, 79, 89, 114\}$    &                &                           &                          &                               &             \\ \hline
  10    & $\{10, 33, 38, 5, 8,$     & $68$      & Yes                    & $\{ 3,4,7,9,10 \}$ & Yes                      & $78$ \\
          & $ 10, 6, 7, 11, 8 \}$       &                &                          &                        &                                   &           \\ \hline        
  10   & $\{10, 33, 40, 5, 8,$      & $69$       & Yes                   & $\{ 3,6,9,10 \}$ & Yes                  & $79$ \\
          & $10, 6, 7, 11, 8 \}$        &                &                           &                        &                                   &           \\ \hline\hline
\end{tabular}
\caption{Experimentation with an Implementation of Algorithm \ref{alg:optimaldynamicprogram} on solving {\bf{PPi}}}\label{table1}
\end{table}

\vspace{10pt}

\noindent We further test our algorithm on the data sets for {\bf{PPi}} found in \cite{Burkardt}.   Results are given in Table \ref{table2}.   As we can see from the table, our algorithm is able to solve all these instances of the partition problem.

\begin{table}[ht]
\centering
\begin{tabular}{||c|c|c|c|c|c|c||}  \hline \hline
Data Set   &  $n$ &     $C$      &        Solution found?      &  $A$                  & Time Taken (sec)     & $C_1^{\ast\ast}(R_1)$ \\ \hline \hline
P01          &  10   &  27           &             Yes         &   $\{4,7,8,9,10\}$    &         0.0064                 &             37                               \\ \hline
P02          &   10 &  2640       &         Yes        &    $\{2,5,7,8,9,10\}$                          &     0.2863       &          2650                                 \\ \hline
P03          &   9 & 1419        &            Yes         &   $\{5,6,9\}$      &    0.0916               &           1428                               \\ \hline
P04          &   5  &  32            &          Yes          &  $\{ 2,4,5 \}$           &       0.0091                          &            37                             \\ \hline
P05          &    9 &  11            &        Yes        &  $\{5,6,7,8,9\}$        &       0.0048                          &          20                               \\ \hline \hline
\end{tabular}
\caption{Solutions to Instances of the Partition Problem in \cite{Burkardt}  using an Implementation of Algorithm \ref{alg:optimaldynamicprogram} }\label{table2}
\end{table}
\vspace{10pt}

\noindent We next report on the time needed to execute the algorithm for different choices of $n$.   We vary $n$ from $20$ to $120$, in intervals of $2$.  For each $n$, we generate random {\bf{PPi}}, with $a_i$, taken from the rounded uniform distribution on $[1,  \lfloor 10000/n \rfloor]$ for $i = 1, \ldots, n-1$, and $10000 - \sum_{k=1}^{n-1} a_k$ for $i = n$.  By doing this, we have $\sum_{i=1}^n a_i = 10000$.  In Figure \ref{figure1}, $t_n$ stands for the time taken to execute the algorithm for a given $n$.  We find the time taken for this using the ``tic'', ``toc'' features in Matlab.  As we can see from the figure, the graph exhibits near constant behavior, with $t_n$ ranging between $\sim 0.57$ sec to $\sim 0.95$ sec, and most values lying between $\sim 0.61$ sec and $\sim 0.77$ sec.  The results seem to indicate no time dependency on $n$.

\begin{figure}[!htpb]
\centering
\includegraphics[width=13.5cm]{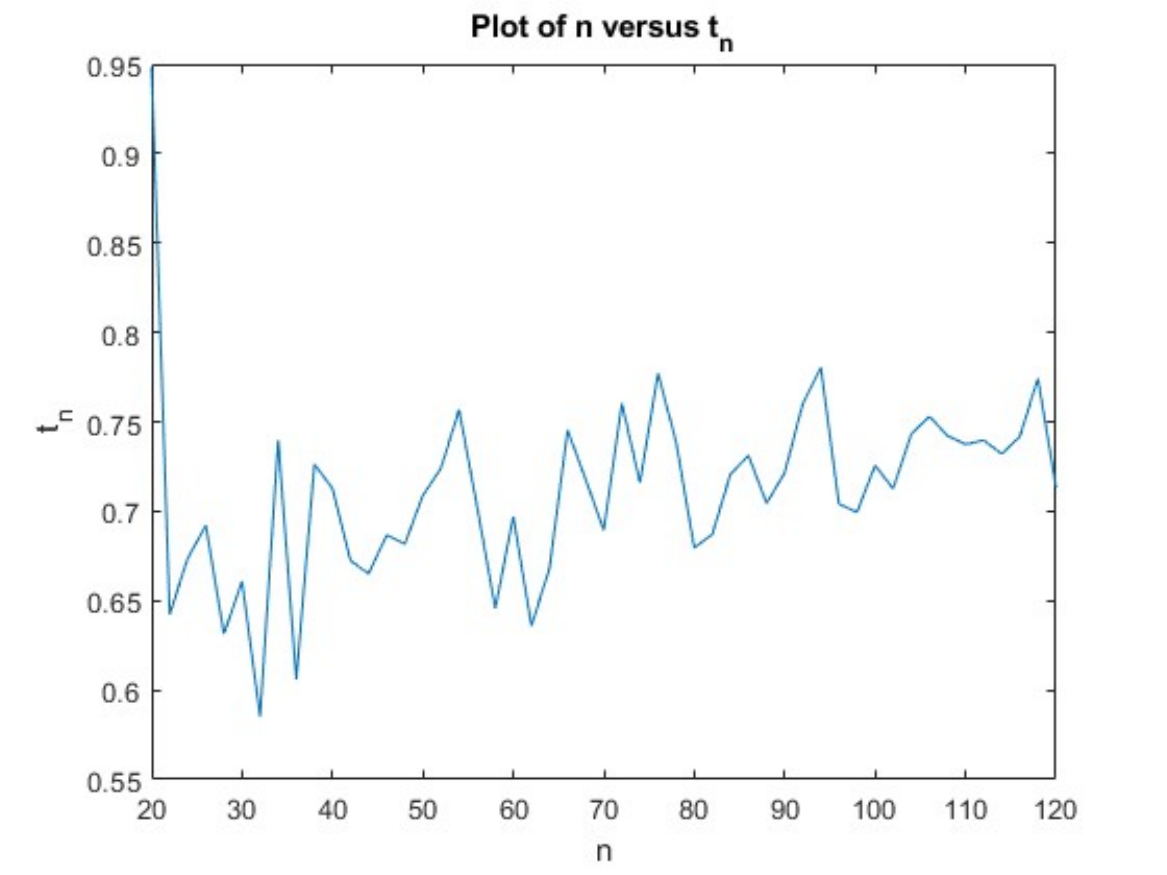}
\centering
\caption{Runtime with $n$} \label{figure1}
\end{figure}

\vspace{10pt}

\noindent We also execute the algorithm for different choices of $C$.  We set $n = 5$ in our experiments.   We let $C_1$ varies from 4000 to 100000 in intervals of 1000.  For each $C_1$, we generate random {\bf{PPi}}, with $a_i$, $i = 1, \ldots, 4$, taken from the rounded uniform distribution on $[1, \lfloor C_1/5 \rfloor ]$, and we set $a_5$ to be $C_1 - \sum_{i=1}^4 a_i$.    Therefore, $\sum_{i=1}^5 a_i = C_1$, and $C = C_1/2$.  We report our findings in Figure \ref{figure2}, where $t_C$ stands for the time taken to execute the algorithm for a given $C$.  We see from the figure that the graph exhibits near linear behavior indicating polynomial time dependency on $C$.


\begin{figure}[!htpb]
\centering
\includegraphics[width=13.5cm]{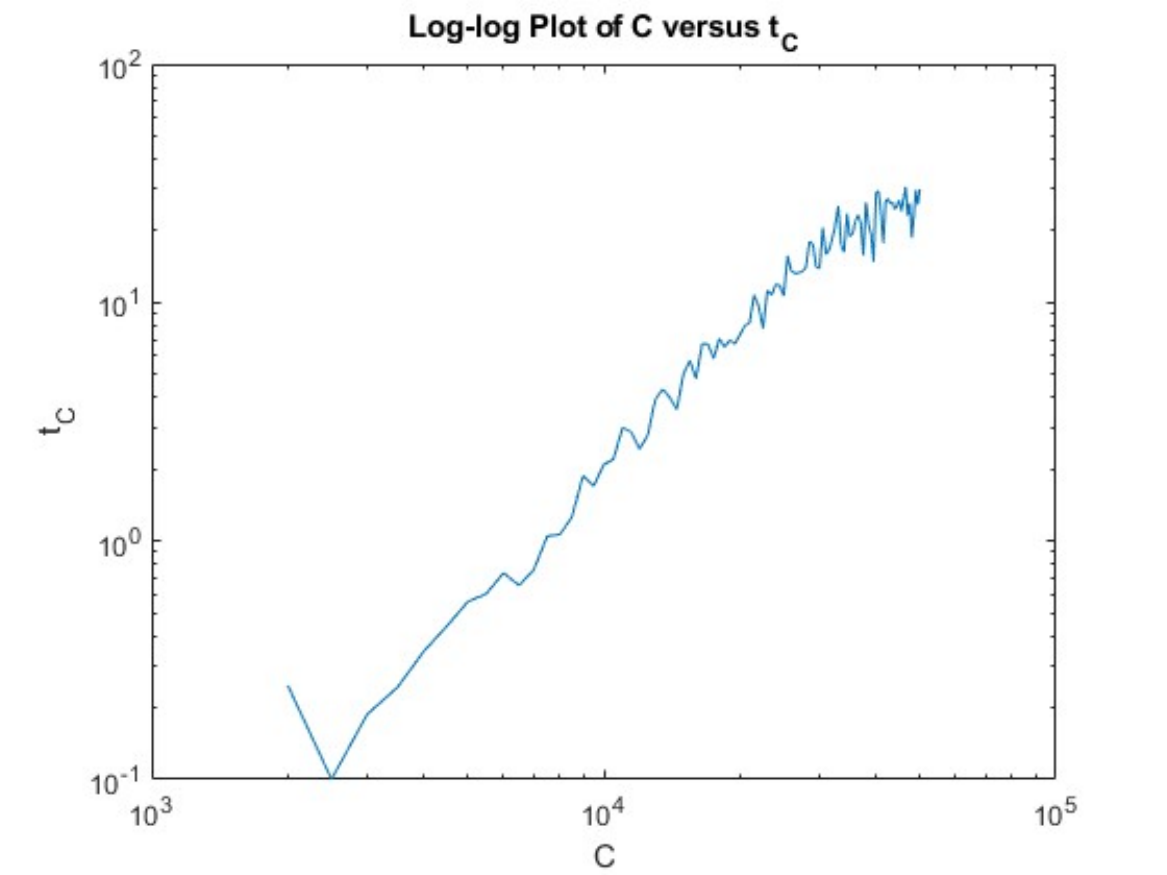}
\centering
\caption{Runtime with $C$} \label{figure2}
\end{figure}

%
%
%
%
%
%
%

\section{A Discussion}\label{sec:Discussion}

\noindent In this section, we discuss the possibility of polynomial time complexity of the partition problem through our approach in this note.  

\vspace{10pt}

\noindent Recall that we formulate a dynamic program for {\bf{DLSP}} in Section \ref{sec:DynamicprogramSolution} in the form of (\ref{eq:recursive}).   To improve the time complexity to solve {\bf{DLSP}}, and in particular, {\bf{DLSPp}}, lies in improving the result in Lemma \ref{lem:dynamicprogramobservation}.   We would require $C_i^{\ast\ast}(J_i)$ to be evaluated for $O(1)$ number of $J_i$, $i = 1, \ldots, N$, which cannot be implied by Lemma \ref{lem:dynamicprogramobservation}, for polynomial time complexity to solve {\bf{DLSPp}}, and hence {\bf{PPi}}.

\bibliographystyle{plain}
\bibliography{Reference_Sim_3}

\end{document}